\documentclass[12pt]{article}
\usepackage{amscd, amsmath}
\usepackage{amsthm, amssymb}
\usepackage[hang,flushmargin]{footmisc}

\newtheorem{theorem}{Theorem}[section]

\newtheorem{proposition}[theorem]{Proposition}

\theoremstyle{remark}
\newtheorem*{remark}{{\bf Remark}}

\newcommand{\CC}{\mathbb{C}}

\newcommand{\dd}{{\rm d}}

\begin{document}

\baselineskip=16pt

\begin{titlepage}
\title{\bf Vortices as degenerate metrics}
\vskip -70pt
%\begin{flushright}
%{\normalsize \ ITFA-2008-21}\\
%\end{flushright}
%\vskip 45pt
%{\bf }
%}
\vspace{3cm}

\author{{J. M. Baptista}}  

\date{April 2013}

\maketitle

\thispagestyle{empty}
\vspace{2cm}
\vskip 20pt
{\centerline{{\large \bf{Abstract}}}}
\noindent
We note that the Bogomolny equation for abelian vortices is precisely the condition for invariance of the Hermitian-Einstein equation under a degenerate conformal transformation. This leads to a natural interpretation of vortices as degenerate hermitian metrics that satisfy a certain curvature equation. Using this viewpoint, we rephrase standard results about vortices and make new observations. We note the existence of a conceptually simple, non-linear rule for superposing vortex solutions, and we describe the natural behaviour of the $L^2$-metric on the moduli space upon restriction to a class of submanifolds.

\let\thefootnote\relax

\end{titlepage}

\section{Introduction}

The abelian Higgs model at critical coupling is a classical field theory defined on a hermitian line bundle $L \rightarrow M$ over a Riemannian manifold. The variables are a $U(1)$-gauge field $A$ and a complex scalar field $\phi$, or, geometrically, a hermitian connection and a smooth section of the bundle. The theory is defined by the static energy functional 
\begin{equation}\label{energy}
E(A, \phi) \ = \ \int_{(M,\,  g)}\,  \frac{1}{2\,e^2}\, \vert F_A \vert^2 \ + \ \vert \dd_A \phi \vert^2  \ + \ \frac{e^2}{2} \, \big( \vert \phi \vert^2 - \tau \big)^2 \ ,
\end{equation}
where $F_A$ is the curvature form, $\dd_A \phi$ is a covariant derivative, and we have explicitly inserted positive constants $e^2$ and $\tau$.  When the background $(M, g)$ is a complex K\"ahler manifold, this functional has the notable property of being minimized precisely by the field configurations $(A, \phi)$ that satisfy a set of first-order PDE's, called the abelian vortex equations \cite{Brad}. Denoting by $\omega$ the K\"ahler form on $(M, g)$, these equations read
\begin{align}
&\bar{\partial}_A \phi \ = \ 0 \    \label{holomorphy_eq1} \\
&i\, \Lambda_\omega F_A \: +\: \, e^2 \, \big(\vert\phi\vert^2 \: - \: \tau \big) \ = \ 0 \  \label{vortex_eq1} \\
& F_A^{0,2} \ = 0 \ , \label{integrability_condition}
\end{align}
where $\Lambda_\omega F_A$ stands for the contraction of the curvature with the K\"ahler form. The equations make sense on any complex Hermitian manifold. Analytically, it is well known that solving the abelian vortex equations is equivalent to solving the Kazdan-Warner equation \cite{Brad}. The latter is a classical, second order PDE, that originally appeared in the problem of finding Riemannian metrics with prescribed scalar curvature \cite{KW}. So in some sense abelian vortices should be related to Riemannian metrics. Our aim here is to formulate this relation in a natural way. 

Roughly speaking, our basic observation is that the vortex equations are precisely the condition for invariance of the form $(i \Lambda_\omega F_\omega + e^2 \tau) \, \omega $ under the degenerate conformal transformation $\omega \rightarrow \tau^{-1} \vert \phi \vert^2 \, \omega$. Here $F_\omega$ denotes the Riemannian curvature form of the K\"ahler metric on $M$, although this same observation can be extended to conformal transformations of hermitian metrics on arbitrary complex vector bundles over $M$ (cf. Theorem \ref{vdm}). This observation has several interesting consequences. The first is that when $M$ is a Riemann surface, the most important case in the physics literature \cite{Manton-Sutcliffe}, studying abelian vortices is essentially equivalent to studying degenerate metrics on $M$ that satisfy a natural curvature equation. In Section 2 we explore this fact and rephrase many of the standard results about vortices on surfaces in terms of degenerate metrics. Arguably, features like the minimal volume bound and the special role of hyperbolic surfaces emerge more naturally from this viewpoint. However, the main new advantage of the curvature equation is that it exhibits two obvious symmetries that are less manifest in the traditional vortex equations. These symmetries imply the existence of a conceptually simple, non-linear rule for superposing vortex solutions, a fact that is physically and mathematically appealing, and that seems not to have been remarked before. As a first application, this superposition rule and the interpretation of vortices as degenerate metrics, taken together, suggest that a system of $d$ vortices living on a surface with one vortex constrained to a fixed position, should, somehow, be  similar to a system of $d-1$ free vortices living on that same surface equipped with an appropriately deformed background metric. Topologically, both systems have the same moduli space. In Section 2 we define what ``appropriately deformed" means and show that, with respect  to the natural $L^2$-metrics, the two moduli spaces are isometric.

In Section 3 we extend the discussion to vortices living on higher-dimensional K\"ahler or Hermitian manifolds. We describe the  relation between the abelian vortex equations and the Hermitian-Einstein equation on complex vector bundles, and we examine how the energy functional transforms under vortex superposition. In a final, brief paragraph, we note that a slightly modified version of the abelian vortex equations, a version equivalent to the perturbed Seiberg-Witten equations in the case of a K\"ahler surface, is in fact just the condition for constant scalar curvature of the deformed metric $ \tau^{-1} \vert \phi \vert^2 \, \omega$. Although we work on compact manifolds, a large part of the discussion is local, so follows through to the non-compact case.

\section{Vortices on Riemann surfaces}

\subsubsection*{The curvature equation}

Let $M$ be a compact Riemann surface equipped with a K\"ahler form $\omega$. Let $L \rightarrow M$ be a hermitian line bundle over that surface and suppose that we are given a connection $A$ and a non-trivial section $\phi$ of $L$ such that $\bar{\partial}_A \phi =  0 \, $.
This condition means that $\phi$ is a holomorphic section with respect to the holomorphic structure on $L$ induced by the connection $A$. In particular, $\phi$ vanishes at isolated points $q_j \in M$ with a well defined, positive multiplicity $n_j \in \mathbb{Z}$. So a solution of \eqref{holomorphy_eq1} determines an effective divisor $D = \sum n_j \, q_j$ on the surface $M$.  Consider now the degenerate K\"ahler form on the surface defined by 
\begin{equation} \label{conf_transformation1}
\omega' \ :=  \ \frac{1}{\tau} \, \vert \phi \vert^2 \ \omega \ .
\end{equation}
This is a smooth 2-form on $M$ that vanishes around each point $q_j$ as $|z|^{2 n_j}$, where $z$ is a complex coordinate on $M$ centered at $q_j$. We will say that a K\"ahler metric with this property is degenerate along the divisor $D$. Since on a surface equation \eqref{integrability_condition} is always satisfied, our basic question is the following: suppose that the pair $(A, \phi)$ satisfies also the vortex equation \eqref{vortex_eq1}. What does this imply for the degenerate form $\omega'$?

\begin{proposition}
The pair $(A, \phi)$ satisfies the vortex equation \eqref{vortex_eq1}
on $M$ if and only if the metric $\omega'$ is degenerate along $D$ and satisfies
\begin{equation} \label{curvature_equation}
i F_{\omega'}\: +\:  e^2 \tau \, \omega' \ = \ i F_{\omega} \: + \:  e^2 \tau \, \omega \ 
\end{equation}
over $M \setminus \{ q_j \}$, where $F_{\omega}$ denotes the curvature form of the K\"ahler metric $\omega$.
\end{proposition}
\begin{proof}
The basic observation is that,  away from the zeroes of $\phi$, the curvatures are related by $F_A = \bar{\partial} \partial \log \vert \phi \vert^2 = F_{\omega'} - F_\omega$. See the proof of Theorem \ref{vdm} for more details.
\end{proof}
\noindent
So the vortex equation is precisely the condition for the conformal transformation \eqref{conf_transformation1} to preserve the form $iF_{\omega} + e^2 \tau \, \omega$. Studying abelian vortices on $(M, \omega)$ is the same thing as studying K\"ahler metrics $\omega'$ that satisfy equation \eqref{curvature_equation}.  Abelian vortices define such degenerate metrics and, conversely, all vortex solutions are given by quotients of the form $\omega' / \omega$.
The well known fact that there exists only one vortex solution $(A, \phi)$ with nowhere vanishing $\phi$, up to gauge transformations, and that this solution has constant norm $\vert\phi\vert^2 = \tau$, just means that $\omega' = \omega$ is the unique non-degenerate solution to the curvature equation \eqref{curvature_equation}. Given an effective divisor $D = \sum n_j \, q_j$ on $M$, the well known fact that there exists only one vortex solution $(A, \phi)$ associated to $D$, up to gauge transformations \cite{Brad, Garcia-Prada, Jaffe-Taubes}, tells us that there exists a unique metric $\omega'$ on the Riemann surface that satisfies \eqref{curvature_equation} and is degenerate at the points $q_j \in M$ with multiplicity $n_j$. 

%Finally, it follows from definition \eqref{conf_transformation1} that the curvatures in the two equations \eqref{vortex_eq1} and \eqref{curvature_equation} are related by $F_A = F_{\omega'} -F_\omega$. So the well known fact that the curvature (or magnetic flux) of a vortex solution is concentrated around the zeros of the Higgs field $\phi$, just means that $F_{\omega'}$ differs the most from $F_\omega$ around the degeneracy points of $\omega'$.

\subsubsection*{Vortices on hyperbolic surfaces}

Suppose that the initial metric $\omega$ on the surface has constant scalar curvature equal to $-e^2 \tau$.
Equation \eqref{curvature_equation} then reduces to $i F_{\omega'} + e^2 \tau \, \omega' = 0$, whose solutions are again hyperbolic metrics. Thus each vortex solution on $(M, \omega)$ defines through \eqref{conf_transformation1} a degenerate hyperbolic metric $\omega'$ with the same curvature $-e^2 \tau$. Conversely, all vortex solutions on $(M, \omega)$ can be obtained as quotients $\tau\, \omega' / \omega$ of hyperbolic metrics. This is a familiar fact  when $M = \mathbb{H}$ is the hyperbolic half-plane. In this case, all degenerate hyperbolic metrics on $\mathbb{H}$ can be obtained as pullbacks $\omega' = f^\ast \omega$ of the Poincar\'e metric by Blaschke products $f: \mathbb{H} \rightarrow \mathbb{H}$. So all vortex solutions can be explicitly written as quotients $(f^\ast \omega) / \omega$, a fact that was originally noted by Witten \cite{Witten}. When $(M, \omega)$ is a compact hyperbolic surface, Manton and Rink noted that one can also obtain vortex solutions as quotients $(f^\ast \omega_Y) / \omega$, where $(Y, \omega_Y)$ is a second hyperbolic surface and $f : M \rightarrow Y$ is a ramified holomorphic map \cite{Manton-Rink}. Unfortunately these maps are rare, so only isolated and non-explicit vortex solutions can be obtained in this way (an upset that improves marginally if we accept vortices with singularities \cite{Baptista-Biswas}). Our observation here is that, even in the compact case, all vortex solutions can still be written as quotients of hyperbolic metrics on $M$. The difficulty is that the degenerate metrics $f^\ast \omega_Y$ provided by ramified maps represent only a tiny subset of all the degenerate hyperbolic metrics on $M$.

\subsubsection*{Degenerate metrics}

A word now about standard properties of degenerate metrics. Firstly, observe that the Riemannian curvature form $F_{\omega'}$ has a smooth extension to the degenerate points, even though the local 1-forms $A_{\omega'}$ of the Chern connection do not. In fact, around each $q_j \in M$ we can write
\begin{equation} \label{degeneracies}
\omega' = \vert z \vert^{2n_j} \, f_j \, \, \omega \ , 
\end{equation}
where $f_j$ is a smooth and positive function locally defined around $q_j$. It follows that the Levi-Civita (or Chern) connection is given by
\[
A_{\omega'} \ = \  A_\omega \: + \:  n_j \, z^{-1} \dd z \: + \: \partial (\log f_j) \ 
\]
in a local holomorphic trivialization of $TM$. So the connection is singular at the degenerate points and the residues $n_j$ can be probed by integrating $A_{\omega'}$ around small circles centered at the $q_j$'s. On the other hand, the curvature 
\[
F_{\omega'} \ = \  \dd A_{\omega'} \ = \ F_\omega \: + \: \bar{\partial} \partial  (\log f_j) \ 
\]
can be  smoothly extended to the degenerate points $q_j$. Note that the last term in the equation above, although locally exact, is not globally exact, since the functions $f_j$ are not globally defined. In fact, one can check that if $\omega'$ is defined by \eqref{conf_transformation1} and the pair $(A, \phi)$ solves the holomorphy equation \eqref{holomorphy_eq1}, then away from the zeroes of $\phi$ we have
\begin{equation} \label{relation_curvatures}
F_{\omega'} \ =\   F_\omega \ + \ F_A \ ,
\end{equation}
where, again, $F_A$ is the curvature of the $U(1)$-connection $A$. This identity is another way to recognize that $F_{\omega'}$ can be smoothly extended to the whole $M$. However, note that this smooth extension is not globally a curvature form, because the Levi-Civita connection has singularities. In particular, one cannot directly apply Stokes' theorem to $F_{\omega'}$ over domains that contain degenerate points. The Gauss-Bonnet theorem in its usual form is also not applicable to the smooth extension of $F_{\omega'}$, as the integral of identity \eqref{relation_curvatures} yields
\begin{equation} \label{gauss_bonnet1}
\frac{i}{2\pi} \: \int_{M\setminus \{ q_j \}} F_{\omega'} \ = \ (2 - 2 g) \: + \:  \sum_j \, n_j \ , 
\end{equation}
where $g$ is the genus of $M$. 
To remedy these undesirable features, one usually considers an extension of the curvature $F_{\omega'}$ to the degenerate points that is different from the smooth one. This second extension regards the curvature of the degenerate metric as an integration current,  and is defined by
\begin{equation}\label{curvature_current}
F_{\omega'} \ := \  \begin{cases}
                                                \  \dd A_{\omega'}    & \text{\rm over } M \setminus \{ q_j \}  \\
                                                 \  2 \pi i \, n_j \, \delta(q_j)  &  \text{\rm at each point } q_j \ ,
                                          \end{cases}           
\end{equation}
where $ \delta(q_j)$ is the delta function. Thus, at the price of considering distributions, the usual formulae are applicable to the curvature, and we have for instance
\[
 \frac{i}{2\pi} \: \int_{M} F_{\omega'} \ = \ (2 - 2 g) \ .
\]
Defining the curvature of a degenerate metric to have delta functions at the points $q_j$ also makes sense if we think of the degenerate metric as a limit of non-degenerate metrics. In this case one can check that the curvature of the non-degenerate metrics increases at $q_j$ as these metrics tend to zero at $q_j$, and the limit will be a delta function. Moreover, since the curvature of all the non-degenerate metrics satisfies Gauss-Bonnet, it is natural to define a limit curvature $F_{\omega'}$ that also satisfies Gauss-Bonnet.
Observe that if we use integration currents on $M$, the curvature equation \eqref{curvature_equation} should be rewritten as
\begin{equation} \label{curvature_equation_deltas}
i F_{\omega'} \:+ \: e^2 \tau \, \omega' \ = \ i F_{\omega} \: +\:  e^2 \tau \, \omega \:  - \: 2 \pi \, \sum_j n_j \, \delta(q_j)  \ ,
\end{equation}
and this is an equation over the whole surface. Comparing with \eqref{curvature_equation}, we see that the delta functions encode the boundary conditions for $\omega'$ at the degenerate points. They play exactly the same role here as in the original Taubes' equation for abelian vortices \cite{Jaffe-Taubes}.

\subsubsection*{Minimal volume bound}

Integrating the curvature equation \eqref{curvature_equation} over $M \setminus \{q_j\}$, or integrating equation \eqref{curvature_equation_deltas} over the whole $M$, we see that each solution $\omega'$ must satisfy
\[
\text{\rm Vol}(M, \omega') \  = \ \text{\rm Vol}( M,  \omega) \  - \  \frac{2 \pi}{e^2 \tau} \, \sum_j \, n_j \ .
\]
Thus each additional degeneracy (i.e. vortex) reduces the volume of the metric $\omega'$ by $2\pi / (e^2 \tau)$. Since any Riemannian metric with a finite number of degenerate points certainly has positive volume, we also   conclude that a necessary condition for the existence of solutions of the curvature equation is that
\[
2 \pi\, \sum_j \, n_j \ < \ e^2 \tau \, \text{\rm Vol}(M, \omega) \ .
\]
This, of course, is just the standard condition for the existence of vortex solutions associated to a non-zero divisor on $M$ \cite{Brad}.

\begin{remark}
The interpretation of vortices as a deformation of the background Riemannian metric satisfies basic intuitive properties. For instance, it is known that the rescaling factor $\tau^{-1} \vert \phi \vert^2$ tends exponentially fast to the vacuum value $1$ as one gets away from the vortex core, at least for vortices on the flat plane \cite{Jaffe-Taubes}. So the background metric $\omega$ is substantially deformed only in a very localized region around the vortices. Moreover, an application of the maximum principle says that the conformal factor is bounded by  $0 \leq \tau^{-1} \vert \phi \vert^2 \leq 1$ over the entire surface $M$ \cite{Jaffe-Taubes, Brad}. Hence the total deformation of the background metric $\omega $ can be measured by  the diference of volumes $\text{\rm Vol}\ (M, \omega) - \text{\rm Vol}\ (M, \omega')$, which is proportional to the vortex number.

\end{remark}

\subsubsection*{Hidden symmetries and vortex superposition}

The interpretation of abelian vortices as degenerate metrics exposes two symmetries of the vortex equations that are less manifest in their usual form. The first is the fact that the curvature equation \eqref{curvature_equation} is essentially symmetric in $\omega$ and $\omega'$. If $\omega'$ is a solution on the background $(M, \omega)$, then $\omega$ is a solution on the background $(M, \omega')$. To interpret this in terms of the usual vortex equations requires working with meromorphic sections $\phi$. The observation is that if $(L, \phi , A)$ is a solution of \eqref{vortex_eq1} on $(M, \omega)$, then $(L^{-1}, \phi^{-1}, -A)$ is a solution on the background $(M, \tau^{-1} \vert \phi \vert^2\, \omega)$. Meromorphic sections appear because the boundary conditions at the degenerate points make the curvature equation not exactly symmetric in $\omega'$ and $\omega$. This is most manifest in the form \eqref{curvature_equation_deltas} of the equation, where an interchange of the two metrics requires an accompanying change of sign of the divisor in order to preserve the equation.

The second hidden symmetry is a transitivity property. If $\omega'$ is a solution of the curvature equation \eqref{curvature_equation} on the background $(M, \omega)$ and $\omega''$ is a solution on the background $(M, \omega')$, then $\omega''$ is a solution on the background $(M, \omega)$. 
This is evident from
\[
iF_{\omega''} \: + \: e^2 \tau \, \omega'' \ = \ i F_{\omega'} \: + \: e^2 \tau \, \omega' \ = \ i F_{\omega} \: + \:  e^2 \tau \, \omega \ .
\]
In the usual formulation of vortices, this transitivity corresponds to the observation that if
 $(L_1, \phi_1 , A_1)$ is a  vortex solution on $(M, \omega)$ and $(L_2, \phi_2 , A_2)$  is a vortex solution on the degenerate background  $(M,\, \tau^{-1} \vert \phi_1 \vert^2 \, \omega)$, then the tensor product of hermitian bundles $(L_2 \otimes L_1,\, \tau^{-1/2} \phi_2 \, \phi_1 ,  \, A_2 +  A_1)$ is a vortex solution on the original background $(M, \omega)$. It is straightforward to check this directly using equation \eqref{vortex_eq1}. The basic identity is 
 \[
i F_{A_2 + A_1} \ + \  e^2 \big(\tau^{-1} \vert \phi_1 \phi_2 \vert^2  - \tau \big) \omega \ = \ 
i F_{A_2}   \ + \  e^2 \big(\vert \phi_2 \vert^2  -  \tau \big) (\tau^{-1} \vert \phi_1\vert^2 \omega ) \ + \  \Big[  i F_{A_1}   \ + \   e^2 \big(\vert \phi_1 \vert^2   -  \tau \big) \omega  \Big] .
 \]
 This property of vortex solutions is very appealing and, apparently, has not been remarked before. It says that we can construct a $(d_1 + d_2)$-vortex solution on $(M,\omega)$ by first obtaining a $d_1$-vortex solution on $(M,\omega)$, then deforming the background $\omega$ to the degenerate metric $\tau^{-1} \vert \phi_1 \vert^2\, \omega$ dictated by that first solution, and finally solving the equation for $d_2$-vortices on the deformed background. We can also do it the other way around, of course, solving for $d_1$-vortices on the background deformed by $d_2$-vortices. This amounts to a conceptually simple, non-linear rule for superposing vortices. Again, it states that instead of na\"ively multiplying vortex solutions obtained on the same background, the correct prescription to superpose vortices implies using a second solution obtained on the background deformed by the first solution. If one knew how to solve the 1-vortex equation on any background, a recursive application of this rule would yield the $d$-vortex solution on any background.

\subsubsection*{Energy functional}

As mentioned in the introduction, a basic fact about the vortex equations is that they minimize the static energy of  the abelian Higgs model at critical coupling, i.e.  they minimize the energy functional \eqref{energy}. One may thus wonder how to rephrase this in terms of degenerate metrics.
Given an effective  divisor $D = \sum n_j q_j$ on a surface $(M, \omega)$, we define the energy density 
of a K\"ahler metric $\omega'$ that degenerates along $D$ by 
\begin{equation} \label{energy_density}
\mathcal{E}(\omega') \ := \ \frac{1}{2e^2}\, |F_{\omega'} - F_\omega|^2\ +\  2\tau \, |\dd \sqrt{\omega' / \omega} \,|^2\ +\ \frac{e^2 \tau^2}{2} \, |\omega' - \omega|^2 \  .
\end{equation}
This a smooth function on $M \setminus \{q_j \}$. It can be rewritten as
\[
\mathcal{E}(\omega') \ = \ \frac{1}{2 e^2}\,  \big|F_{\omega'} - F_{\omega} - i e^2 \tau (\omega' - \omega) \big|^2 \ + \ i \tau \ast \big[  F_{\omega'} - F_{\omega}  +  \partial \bar{\partial}(\omega' / \omega) \, \big] \ 
\]
over that domain, where $\ast$ is the Hodge operator with respect to $\omega$. Using equation \eqref{gauss_bonnet1}
and Stokes'  theorem, it is clear that the total energy functional satisfies
\[
E(\omega') \ := \ \int_{M \setminus \{q_j\}} \mathcal{E}(\omega')\, \omega\  \geq  \  2 \pi \tau \sum_j\, n_j \ , 
\]
with the equality standing precisely when $\omega'$ satisfies the curvature equation \eqref{curvature_equation} over $M \setminus \{q_j\}$. This is the analog of the usual Bogomolny argument for vortices.

\begin{remark}
At first sight one could be worried that the square root term appearing in the energy density \eqref{energy_density} may lead to singularites at the points where $\omega'$ vanishes. However, a direct calculation shows that if  $\omega'$ vanishes at $q_j$ as $|z|^{2n_j}$, the squared norm $|\dd \sqrt{\omega' / \omega} \,|^2$ is smooth at that point for any integer $n_j \geq 0$.
\end{remark}

\begin{remark}
In this discussion we have fixed the degeneration divisor $D$, have only considered metrics $\omega'$ with the boundary conditions prescribed by $D$, and the integrals are over the domain $M \setminus \text{supp}\, D$. An equivalent approach would be to consider the curvature $F_{\omega'}$ as an integration current over the entire $M$, as in \eqref{curvature_current}, and then encode the boundary conditions in  delta functions appearing in the energy density. In this case one should substitute the term $F_{\omega'} - F_{\omega}$ in the formulae above by $F_{\omega'} - F_{\omega} + 2\pi \delta (D)$ and perform the integrations over the whole $M$.
\end{remark}

\subsubsection*{Metric on the moduli space}

It is well known that for each effective divisor $D = \sum n_j q_j$ on the surface $M$ there is exactly one solution of the vortex equations \eqref{holomorphy_eq1} and \eqref{vortex_eq1}, up to $U(1)$-gauge transformations, such that $D$ is the zero-set divisor of the Higgs field $\phi$. So if we fix the degree of $D$ to be $d$ and let the points $q_j$ vary in $M$, we get a moduli space $\mathcal{M}^d$ of vortex solutions that is isomorphic to the symmetric product $\text{\rm Sym}^d M$ of the surface. In terms of degenerate metrics, $\mathcal{M}^d$ is the space of solutions of the curvature equation \eqref{curvature_equation} such that the K\"ahler metric $\omega'$ has $d$ degenerate points on $M$, counting multiplicities.

The space $\mathcal{M}^d$ has a natural K\"ahler metric induced by its interpretation as a vortex moduli space \cite{Manton-Sutcliffe}. Using the language of \cite[Sect.7]{Bap2010}, the corresponding K\"ahler form $\omega_{\mathcal{M}}$ can be written in terms of the degenerate metrics as
\begin{equation} \label{vortex_metric}
\omega_{\mathcal{M}} \ = \ \frac{-1}{2\, e^2} \, \Big\{  \int_M   \, \partial_z  \Big[ \: \big(\partial_{\bar{u}^k} \log f \big) \: \big(\partial_{\bar{z}} \partial_{u^j} \log f \big)   \: \Big] \:  \dd z \wedge \dd \bar{z}   \: \Big\} \ \dd u^j \wedge \dd \bar{u}^k  \ ,
\end{equation}
where $f (\omega'  , p) :=  (\omega'  / \omega) (p)$ is a non-negative smooth function on the cartesian product $\mathcal{M}^d \times M$, and we chose arbitrary local complex coordinates $z$ on $M$ and $\{u^k\}$ on $\mathcal{M}^d$. Observe that in \eqref{vortex_metric} the integrand 2-form is exact outside the points where $\log f$ explodes, i.e. outside the degenerate points of the metric $\omega'$. By Stokes' theorem, the integral therefore localizes to a sum of integrals around small circles centered at these degenerates points, which of course is just the standard Samols' localization for vortices \cite{Samols, Bap2010}. Note also that due to the presence of the partial derivatives $\partial_{\bar{u}^k}$ and $\partial_{u^j} $, the form $\omega_\mathcal{M}$ does not depend explicitly on the background metric $\omega$ on the surface, though it does depend implicitly through the solutions $\omega'$.

Formula \eqref{vortex_metric} and the superposition rule for vortices imply that the metric $\omega_{\mathcal{M}}$ on the moduli space has an interesting property upon restrictions, which we now explain. Let $\mathcal{M}_p^d \subset \mathcal{M}^d$ be the complex submanifold determined by fixing one vortex at position $p \in M$, i.e. $\mathcal{M}_p^d$ is the subset of unordered multiples in $\text{\rm Sym}^d M$ of the form  $(x_1, \ldots , x_{d-1}, p)$ with varying $x_k \in M$. As a complex manifold, $\mathcal{M}_p^d$ clearly is isomorphic to $\mathcal{M}^{d-1}$. However, the K\"ahler metric  on $\mathcal{M}_p^d$ induced by the embedding in $(\mathcal{M}^d, \omega_{\mathcal{M}})$ is different from the natural metric on $\mathcal{M}^{d-1}$ regarded as the moduli space of $d-1$ vortices living on $(M, \omega)$. In other words, the remaining $d-1$ vortices ``feel" the presence of the vortex fixed at $p$. But how, precisely, does this presence translate into the metric on $\mathcal{M}_p^d$? 
%Using the geodesic approximation to relate the metric on the moduli space to low energy soliton dynamics, we have just observed that the dynamics of $d-1$ free vortices on $M$ is different from the dynamics of $d-1$ vortices on $M$ with an extra vortex fixed at the point $p$. In other words, in their low energy motions, the $d-1$ vortices "feel" the presence of the vortex fixed at $p$. But how, precisely, does this presence affect their motion? 
In view of superposition rule discussed before, a natural guess is that the vortex fixed at $p$ deforms the background metric $\omega$ on the surface $M$ to another metric $\omega''$, a metric that is degenerate at $p$, and then the remaining $d-1$ vortices behave as free vortices on the surface $(M, \omega'')$. To be more precise:

\begin{theorem}
Let $\omega''$ be the unique solution of the curvature equation \eqref{curvature_equation} on $(M, \omega)$ with a single simple degeneracy at the point $p \in M$. Then the submanifold $\mathcal{M}_p^d$  equipped with the metric inherited from $(\mathcal{M}^d, \omega_{\mathcal{M}})$ is isometric to the moduli space $\mathcal{M}^{d-1}$ of $d-1$ vortices living on the background surface $(M, \omega'')$. 
\end{theorem}
\begin{proof}
Take a local chart $\{u^k \}$ of the moduli space $\mathcal{M}^d$ adapted to the submanifold $\mathcal{M}_p^d$. This means that in this chart $\mathcal{M}_p^d$ is determined by the equation $u^d = 0$ and  $\{u^1 , \ldots , u^{d-1}\}$ are good complex coordinates on $\mathcal{M}_p^d$. For any solution $\omega' \in \mathcal{M}^d_p$ of the curvature equation on $(M, \omega)$, we can always factorize
\[
\frac{\omega'}{\omega} \ = \ \frac{\omega'}{\omega''} \: \frac{\omega''}{\omega} \ .
\]
The crucial observation is that the factor $\omega'' / \omega$ is constant on the submanifold $\mathcal{M}_p^d$, it does not depend on the positions of the remaining $d-1$ vortices. So for the coordinates $\{u^1 , \ldots , u^{d-1}\}$ we have
\[
\partial_{u^j} \log (\omega' / \omega) \ = \ \partial_{u^j} \log (\omega' / \omega'') \ 
\]
over $\mathcal{M}_p^d$, and the same applies to the anti-holomorphic derivatives. Thus the restriction of the K\"ahler form \eqref{vortex_metric} to the submanifold $\mathcal{M}_p^d$ is given by the same formula \eqref{vortex_metric}, except that we only use the coordinates $\{u^1 , \ldots , u^{d-1}\}$ and the function $f$ is now $\omega' / \omega''$. This coincides with the formula for the K\"ahler form on the moduli space $\mathcal{M}^{d-1}$ of $d-1$ vortices living on the background surface $(M, \omega'')$.
\end{proof}

There is an obvious extension of this property to the lower dimensional submanifolds $\mathcal{M}_{p_1,  \ldots , p_l}^d$ of the moduli space $\mathcal{M}^{d}$ determined by fixing the positions of $l < d$ vortices. It says that  $\mathcal{M}_{p_1,  \ldots , p_l}^d$ equipped with the metric inherited from $(\mathcal{M}^d, \omega_{\mathcal{M}})$ is isometric to the moduli space $\mathcal{M}^{d-l}$ of $d-l$ vortices living on the background surface $(M, \omega'')$, where this time $\omega''$ is the unique solution of the curvature equation on $(M, \omega)$ with simple degeneracies at the points $p_1, \ldots , p_l$.

\subsubsection*{Vortices with parabolic singularities}

So far we have only considered degenerate metrics $\omega'$ on $M$ that vanish with a positive integral power $\vert z \vert^{2n_j}$ at the degenerate points $q_j$. However, pretty much the whole story can be extended to metrics that vanish as $\vert z \vert^{2 \alpha_j}$ for real numbers $\alpha_j > 0$. These correspond to abelian vortices on $(M, \omega)$ with parabolic singularities.

Suppose that we are given a divisor $\sum \alpha_ j \, q_j$ on $M$ with positive real coefficients. Denote by $[\alpha_j] \in \mathbb{Z}$ the integral part of $\alpha_j$, and call $L \rightarrow M$ the complex line bundle of degree $\sum_j [\alpha_ j]$.  Let $h$ be a hermitian metric on $L$ with parabolic singularities of weight $\alpha_ j - [\alpha_ j]$ at the points $q_j$. This means that if $s$ is a smooth local section of $L$ that does not vanish at $q_j$, the squared norm of $s$ is of the form
\[
|s|_h^2 \ = \ |z|^{2(\alpha_ j - [\alpha_ j])} \, f_j (z, \bar{z}) \  
\]
for some smooth and positive function $f_j$ defined around $q_j$. So $h$ is degenerate at the points $q_j$, and all unitary connections on $(L,h)$ will have logarithmic singularities at those points. However, this degeneracy does not affect the existence of vortex solutions. There still exists a smooth section $\phi$ of $L$ and a unitary connection $A$ that solve the vortex equations \eqref{holomorphy_eq1} and \eqref{vortex_eq1} over $M\setminus \{q_j\}$, and such that $\phi$ vanishes precisely at points $q_j$ with multiplicity $[\alpha_j]$. Moreover, this pair $(A, \phi)$ is unique up to $U(1)$-gauge transformations on $M$ \cite{Biquard-Garcia-Prada}. Observe that for this solution the squared norm $\vert \phi \vert^2_h$ vanishes as $|z|^{2 \alpha_ j}$ at the point $q_j$. Therefore the Riemannian metric $\omega':= \tau^{-1} \vert \phi \vert^2_h \, \omega$ satisfies the curvature equation \eqref{curvature_equation} over $M\setminus \{q_j\}$ and is degenerate along the divisor $\sum \alpha_ j q_j$, as required. All in all, we see that studying vortices on $(M, \omega)$ with parabolic singularities is the same thing as studying solutions $\omega'$ of \eqref{curvature_equation} that vanish with positive real weights. All the formulae presented above remain valid if we substitute the integers $n_j$ by the real weights $\alpha_j$.

\begin{remark}
The existence and uniqueness of vortex solutions, and hence of solutions to \eqref{curvature_equation} with prescribed degeneracies,  has also been extended to the case where the background metric $\omega$ is allowed to have degeneracies or conical singularies, meaning that at a finite number of points $p_k \in M$, possibly coincident with the $q_j$'s,  the background K\"ahler form $\omega$ behaves as $\vert z \vert^{2\beta_k}$ with real weights $\beta_k > -1$ \cite{Baptista-Biswas}.
\end{remark}

\section{Vortices in higher dimensions}

In the previous Section we have seen how the vortex equation \eqref{vortex_eq1} can be regarded as the condition for the conformal transformation \eqref{conf_transformation1} to preserve the form $i F_\omega + e^2 \tau \omega$. Observe that the metric $\omega$ plays here a double role, both as a Riemannian metric on $M$ and as a hermitian metric on $TM$ with curvature $F_\omega$. If we separate these roles and substitute $TM$ by an arbitrary holomorphic vector bundle, it becomes clear that the vortex equation is in fact the condition for invariance of the general Hermitian-Einstein equation under a  conformal transformation.
Moreover, this is true on any Hermitian manifold $M$, not just on Riemann surfaces.

Let $(M, \omega)$ be a compact Hermitian manifold and let $L \rightarrow M$ be a hermitian line bundle. Suppose that we are given a connection $A$ and a non-trivial section $\phi$ of $L$ such that
\begin{equation} \label{holomorphy_condition2}
\bar{\partial}_A \phi \ = \ 0 \ , \qquad \quad F_A^{0,2} \ = \ 0 \ .
\end{equation}
Then $\phi$ is a holomorphic section with respect to the holomorphic structure on $L$ induced by the connection $A$. In particular, its zero set determines an effective divisor $D = \sum_j n_j \, Z_j$ of analytic hypersurfaces on $M$. The squared norm $\vert \phi \vert^2$ is a smooth function on M that vanishes along the support of $D$. Now let $V \rightarrow M$ be a holomorphic vector bundle equipped with a hermitian metric $f$. We can consider the double rescaling
\begin{equation} \label{conf_transformation2}
f' \ := \ \frac{1}{\tau'} \, |\phi|^{2\alpha} \ f \qquad \text{and} \qquad  \omega' \ :=  \ \frac{1}{\tau} \, \vert \phi \vert^2 \ \omega \ 
\end{equation}
for any exponent $\alpha \neq 0$.
Then $\omega'$, for instance, is a smooth 2-form on $M$ that vanishes at each hypersurface $Z_j$ as $|z|^{2 n_j}$, where $z$ is a complex coordinate on $M$ such that $Z_j$ is locally defined by the equation $z=0$. In the language of Section 2, the degeneracies of $\omega'$ and $f'$ are prescribed by the divisors $D$ and $\alpha D$, respectively. Once again, our question is the following: suppose that the pair $(A, \phi)$ satisfies the vortex equation \eqref{vortex_eq1}. What does this imply for the degenerate metrics $f'$ and $\omega'$?
\begin{theorem} \label{vdm}
The pair $(A, \phi)$ satisfies the vortex equation \eqref{vortex_eq1}
on $M$ if and only if the metrics $\omega'$ and $f'$ defined by the conformal transformation \eqref{conf_transformation2} satisfy the equation
\begin{equation} \label{curvature_equation2}
\big(\, i \Lambda_{\omega'} F_{f'}\: + \: e^2\, \alpha \, \tau \, \big) \, \omega' \ = \ \big(\, i \Lambda_\omega F_{f} \: + \: e^2\, \alpha\, \tau  \, \big)\, \omega \ 
\end{equation}
over $M \setminus \cup_j Z_j$, where $F_{f}$ denotes the curvature of the Chern connection on $(V, f)$.
\end{theorem}
\noindent
In other words, the abelian vortex equation \eqref{vortex_eq1} is precisely the condition for invariance of the form 
$\big(\, i \Lambda_\omega F_{f} \: + \: e^2 \alpha \, \tau  \, \big)\, \omega$
 under the degenerate conformal transformation \eqref{conf_transformation2}.  The existence of smooth vortex solutions has been proved by Bradlow in the case of K\"ahler $(M, \omega)$ \cite{Brad}, and by Lupascu in  the case of Hermitian manifolds \cite{Lupascu}.

 \begin{remark}
 Choose $V=TM$, $f = \omega$ and $\alpha =1$.  Observe that  
 if the initial metric $\omega$ satisfies the Hermitian-Einstein condition $i \Lambda_\omega F_{\omega}  + e^2  \tau \text{I} = 0$, then each vortex solution on $(M, \omega)$ gives us a new degenerate metric $\omega'$ on $M$ that satisfies the same condition. Taking the trace in equation \eqref{curvature_equation2}, we see that if the initial metric $\omega$ has constant scalar curvature $s_\omega =  \text{Tr}(i\Lambda_\omega F_{\omega}) = - e^2 \tau \dim_\CC M$ (not necessarily being Hermitian-Einstein),  then each vortex solution on $(M, \omega)$ gives us a new degenerate metric $\omega'$ with the same scalar curvature.
 However, note that if $(M, \omega)$ is a higher-dimensional K\"ahler manifold, the rescaled forms $\omega'$ are no longer $\dd$-closed for non-trivial $\phi$, and hence will only define Hermitian metrics on $M$.
\end{remark}

\begin{remark}
Just as in the case of Riemann surfaces, if we regard the curvature $F_{f'}$ of a degenerate metric as an integration current, the domain of equation \eqref{curvature_equation2} can be extended to the entire manifold $M$. In this case one should add a delta function term the curvature equation, a term that encodes the boundary conditions in \eqref{curvature_equation2}. The equation over $M$ then becomes
\begin{equation} \label{curvature_equation_deltas2}
\big(\, i \Lambda_{\omega'} F_{f'}\: + \: e^2\, \alpha \, \tau \, \big) \, \omega' \ = \ \big(\, i \Lambda_\omega F_{f} \: + \: e^2\, \alpha\, \tau  \, \big)\, \omega \  - \ 2 \pi \, \delta_D \ ,
\end{equation}
where $\delta_D$ denotes the current of integration over the divisor $D$. That this is the correct term to add, is a consequence of the Poincar\'e-Lelong formula
$
\bar{\partial} \partial (\log \vert \phi \vert^2 ) \ = \ F_A \ + \ 2\pi i \delta_D$ 
applied to the line bundle $L$.
\end{remark}

\begin{proof}[{\bf Proof of Theorem \ref{vdm}}]
Let $h$ denote the background hermitian metric on $L$. Since the pair $(A, \phi)$ satisfies equations \eqref{holomorphy_condition2}, there exists a holomorphic structure on $L$ such that $A$ coincides with the Chern connection and $\phi$ is a holomorphic section. In a local holomorphic trivialization of $L$, the curvature of the Chern connection $A$ is given by $F_A = \bar{\partial} \partial (\log h)$. Since $\phi$ is holomorphic, away from its  vanishing set we have that $\bar{\partial} \partial \log ( \phi \bar{\phi}) = 0$. Hence on $M \setminus \cup_j Z_j$ the curvature is given by $F_A = \bar{\partial} \partial (\log \vert \phi \vert^2_h)$. On the other hand, in a local holomorphic trivialization of the vector bundle $V$, the curvature of the Chern connection is given by the analogous formula $F_f = \bar{\partial} [ (\partial f) f^{-1} ]$, where the hermitian metric $f$ should be locally regarded as a square matrix of dimension equal to the rank of $V$ \cite[p. 256]{Chern-Chen-Lam}. In particular
\begin{align*}
F_{f'} \: - \: F_f \ &= \  \bar{\partial} \big[  \, (\partial f') (f')^{-1} \, - \,  (\partial f) f^{-1} \,  \big] \ = \ 
 \bar{\partial} \big[\, (\partial f') f^{-1} (f' f^{-1})^{-1} \, + \, (f' f^{-1})^{-1} f' (\partial f^{-1}) \, \big]  \\
 &= \ \bar{\partial} \big\{ \, [ \, (\partial f') f^{-1} \, +\, f' (\partial f^{-1}) \,  ] \, (f' f^{-1})^{-1} \, \big\} \ = \ 
 \bar{\partial} \partial \, \log\, (f' f^{-1}) \ = \ \alpha \, F_A \, \text{I} \  ,
\end{align*}
where $\text{I}$ stands for the identity automorphism of $V$ and we have used that $f' f^{-1} = \vert \phi \vert^{2 \alpha} \, (\tau')^{-1} \, \text{I}$ commutes with all matrices. Thus on $M \setminus \cup_j Z_j$ we have that 
\[
\Big[ \: i \Lambda_\omega  F_A  \: + \:  e^2 \big(   \vert \phi \vert^2 - \tau  \big)  \Big] \, \alpha \,  \omega \ = \  i (  \Lambda_{\omega'}  F_{f'})\, \omega' \: - \:  i (  \Lambda_{\omega}  F_{f})\, \omega \: + 
\:  e^2 \, \alpha \, \tau (\omega' - \omega) \ ,
\]
and the result follows.
\end{proof}

\subsubsection*{Vortex superposition and the energy functional}

The curvature equation \eqref{curvature_equation2} has the same symmetry and transitivity properties that were discussed in the case of vortices on Riemann surfaces. It follows that the superposition rule of Section 2 is still valid for abelian vortices over higher-dimensional $M$. In the traditional form of the vortex equations, this follows from the Leibniz rule
\[
\bar{\partial}_{A_2 + A_1} ( \tau^{-1/2}\, \phi_2 \, \phi_1) \ = \   \tau^{-1/2}\, \phi_2\, \bar{\partial}_{A_1} \phi_1
\ + \ \tau^{-1/2}\, \phi_1\, \bar{\partial}_{A_2} \phi_2  \ , 
\]
from the $A$-linearity of the curvature of abelian connections
\[
F_{A_2 + A_1}^{\, 0,2} \ = \ F_{A_2 }^{\, 0,2} \ + \ F_{A_1 }^{\, 0,2} \ ,
\]
and from the identity
\begin{align*}
\Big[ i \Lambda_{\omega} F_{A_2 + A_1} \ + \  e^2 \big(\tau^{-1} \vert \phi_1 \phi_2 \vert^2 \, - \, \tau \big)  \Big] \omega \ = \ 
  \Big[ i \Lambda_{\omega'} F_{A_2}   \ &+ \  e^2 \big(\vert \phi_2 \vert^2  \, - \, \tau \big)  \Big]\omega'  \ + \\
  &+\  \Big[  i \Lambda_\omega F_{A_1}   \ + \   e^2 \big(\vert \phi_1 \vert^2  \, - \, \tau \big)  \Big] \omega \, ,
\end{align*}
where $\omega'$ denotes the deformed hermitian metric $\tau^{-1} \vert \phi_1\vert^2 \omega$.  So if $D_1$ and $D_2$ are two effective divisors on a Hermitian manifold $(M,\omega)$,  we can construct the vortex solution associated to $D_1 + D_2$ by first obtaining the solution $(A_1, \phi_1)$ on $(M,\omega)$ associated to $D_1$, then deforming the background $\omega$ to the degenerate metric $\omega' = \tau^{-1} \vert \phi_1 \vert^2\, \omega$ dictated by that first solution, and finally obtaining the $D_2$-vortex solution on the new background.

One may wonder if this superposition rule for BPS vortices could be a reflex of  a stronger structural symmetry of the theory, for instance a decomposition rule of the energy functional when evaluated on composed fields like $(A_1 + A_2, \, \tau^{-1/2} \phi_1 \phi_2)$. Unfortunately such a strong symmetry does not seem to hold in general.  However, it may be worth to present a few observations in this direction.  Start by recalling the energy functional of the abelian Higgs model defined in \eqref{energy}.
A version of the  Bogomolny argument \cite{Brad}  says that this functional can be rewritten as
\begin{equation}\label{bogomolny_argument}
E(A, \phi, \omega) \ = \  T(A, \phi , \omega ) \ + \ \int_M {\mathcal{E}}(A, \phi, \omega) \  ,
\end{equation}
where we have separated the non-negative density
\[
{\mathcal{E}}(A, \phi, \omega) \ :=  \  \Big\{ \ \frac{1}{2e^2} \, \big\vert i \Lambda_\omega F_A  + e^2 \big( \big\vert \phi \big\vert^2 - \tau \big) \big\vert^2_\omega \ + \ 2\,  \big\vert \bar{\partial}_A \phi \big\vert^2_\omega \ +\ \frac{2}{e^2} \, \big\vert F_A^{0,2}  \big\vert^2_\omega \ \Big\} \  \frac{\omega^{m}}{m!}
\]
from the remaining integral
\[
T(A, \phi, \omega) \ :=  \ \int_M \ i \tau\, F_A \wedge \frac{\omega^{m-1}}{(m-1)!} \ +\ \frac{1}{2e^2}\, F_A \wedge F_A \wedge \frac{\omega^{m-2}}{(m-2)!} \ - \  i \dd \langle \phi, \dd_A \phi \rangle  \wedge \frac{\omega^{m-1}}{(m-1)!}
 \ .
 \]
(Here the inner product $\langle \phi, \dd_A \phi \rangle =  h \, \bar{\phi}\, \dd_A \phi$ uses the background hermitian metric $h$ on $L$.) The non-negative density ${\mathcal{E}}$ vanishes precisely when the vortex equations are satisfied. The integral $T$ is independent of the fields $(A, \phi)$ whenever the form $\omega$ is closed, i.e. whenever the background $(M, \omega)$ is K\"ahler.  The well known conclusion is that for K\"ahler $\omega$ the field configurations that minimize the total energy $E(A, \phi, \omega)$ are precisely the vortex solutions. For a general Hermitian background, a priori this need not be true.

Our first observation follows from a direct computation using the definition of ${\mathcal{E}}$.

\begin{proposition}
Let $(A_1, \phi_1)$ be a vortex solution on a Hermitian manifold $(M, \omega)$ and, as usual,  let $\omega'$ denote the deformed metric $\tau^{-1} \vert \phi_1 \vert^2\, \omega$. Then we have that
\begin{equation} \label{density_decomposition}
{\mathcal{E}}(A + A_1, \,  \tau^{-1/2} \phi \, \phi_1 , \,\omega) \ = \  \big(\tau^{-1} \vert \phi_1 \vert^2 \big)^{2-m} \   {\mathcal{E}}(A, \phi, \omega')
\end{equation}
for arbitrary fields $(A, \phi)$. 
\end{proposition}
\noindent
This is a remarkably simple transformation rule. In particular, if the fields $(A, \phi)$ are a vortex solution on $(M, \omega')$, the density ${\mathcal{E}}(A, \phi, \omega')$ vanishes everywhere, and so does the density on the left-hand-side of \eqref{density_decomposition}, which in turn means that the superposed fields $(A + A_1, \,  \tau^{-1/2} \phi \, \phi_1)$  satisfy the vortex equations on $(M, \omega)$. So the vortex superposition rule can be directly read from  \eqref{density_decomposition}.

Moreover, notice the special stand of complex dimension $m=2$ in identity \eqref{density_decomposition} . In this dimension the operation of superposing a vortex solution to arbitrary fields can be entirely absorbed in a deformation of the background metric $\omega$, as far as the density ${\mathcal{E}}$ is concerned. Unfortunately the same statement does not hold for the full energy functional $E(A, \phi, \omega)$. The reason is that the deformed metric $\omega'$ is not K\"ahler, and hence the term $T(A, \phi, \omega')$ is no longer topological, i.e. it depends on the fields $(A, \phi)$. This means that the energy $E$ is no longer the same, up to a constant, as the integral of ${\mathcal{E}}$. However, one can define a ``corrected" energy $\hat{E}$ that does transform nicely upon composition with vortex solutions. Still in complex dimension $m=2$, put
\begin{equation*}
\hat{E}(A, \phi, \omega) \ := \  E(A, \phi, \omega) \ + \ \int_M -i \tau\, F_A \wedge \omega\  + \  i\, \dd \langle \, \phi , \dd_A \phi \rangle \wedge \omega \ .
\end{equation*}
Bearing in mind identities \eqref{bogomolny_argument} and \eqref{density_decomposition}, it is manifest that:
\begin{itemize}
\item[\bf (i)] On a fixed compact K\"ahler background $(M, \omega)$, the functionals $E$ and $\hat{E}$ differ only by a  topological constant.
\item[\bf (ii)]  On a general compact Hermitian background $(M, \omega)$, given a vortex solution $(A_1, \phi_1)$, the functional $\hat{E}$ satisfies
\begin{equation}
\hat{E}(A + A_1, \,  \tau^{-1/2} \phi \, \phi_1 , \,\omega) \ =  \   \hat{E}(A, \phi, \omega')
\end{equation}
for arbitrary fields $(A, \phi)$.
\end{itemize}
This means that in the classical theory defined by $\hat{E}$, composing the fields with a vortex solution is fully equivalent to deforming the background metric.
%In particular, it follows from $(i)$ that on a K\"ahler background the Euler-Lagrange equations of $E$ and $\hat{E}$ are identical. It follows from $(ii)$ that  if $(A, \phi)$ solves the Euler-Lagrange equations for the functional $\hat{E}$ on the deformed background $\omega'$, then $(A + A_1, \,  \tau^{-1/2} \phi \, \phi_1)$ solves the Euler-Lagrange equations for $\hat{E}$ on $(M, \omega)$.

\subsubsection*{The modified vortex equation}

As before, let $L \rightarrow M$ be a hermitian line bundle and let $(A, \phi)$ be a pair that satisfies the holomorphy equations \eqref{holomorphy_condition2}. We assume that $\phi$ is non-trivial, so that its zero set determines an effective divisor $D = \sum_j n_j \, Z_j$ on $M$.  A modified version of the vortex equation has been often considered, namely
\begin{equation} \label{vortex_eq3}
i \Lambda_\omega F_A \: + \:  \frac{1}{\dim_\CC M} \, \Big( \,\vert \phi \vert^2 \: + \: s_\omega \, \Big) \ = \ 0  \ ,
\end{equation}
where $s_\omega = i \, \text{Tr}( \Lambda_\omega F_\omega)$ denotes the scalar curvature of $(M, \omega)$. When $M$ is a K\"ahler surface, this equation is equivalent to a perturbation of the Seiberg-Witten equations \cite{Bradlow-Garcia-Prada}. For compact K\"ahler $(M, \omega)$, it can be shown that the existence of solutions to \eqref{vortex_eq3} depends only on the average value 
\[
\bar{s}_\omega \ := \ \frac{1}{\text{\rm Vol}\, (M, \omega)}\,  \int_{(M, \omega)} s_\omega
\]
of the scalar curvature of $\omega$. Solutions $(A, \phi)$ of \eqref{vortex_eq3} exist if and only if the vortex equation \eqref{vortex_eq1} with $\tau = - \bar{s}_\omega$ has solutions \cite{Bradlow-Garcia-Prada}. In particular, we must have $\bar{s}_\omega < 0$ for non-trivial solutions to exist. For a general compact Hermitian manifold $(M, \omega)$, the existence condition is slightly more evolved \cite{Lupascu}. Equation \eqref{vortex_eq3} also has a very simple interpretation in terms of degenerate metrics. As usual, for given data $(L, M, \omega, A, \phi)$, define the deformed metric
\[
\omega' \ :=  \ \frac{- 1}{\bar{s}_\omega}\, \vert \phi \vert^2  \ \omega \ .
\]
\begin{proposition}
The pair $(A, \phi)$ satisfies the modified vortex equation \eqref{vortex_eq3} on $M$
 if and only if the degenerate metric $\omega'$ has constant scalar curvature $s_{\omega'} \: = \: \bar{s}_\omega $ over $M \setminus \cup_j Z_j$.
\end{proposition}
\begin{proof}
If $K \rightarrow M$ denotes the canonical line bundle, the induced hermitian metric on $K$ rescales as
\[
\det \omega' \ =\   \Big( \frac{-1}{\bar{s}_\omega} \, \vert \phi \vert_h^{2} \Big)^{\dim_\CC M} \: \det \omega \ . 
\]
Using that $F_A = \bar{\partial} \partial (\log \vert \phi \vert^2)$ away from the vanishing set of $\phi$, as before, it is clear that the real Ricci form transforms as
\[
\rho_{\omega'} \ =\  i  \bar{\partial} \partial \log (\det \omega')  \ = \  i\, (\dim_\CC M) \, F_A \: + \: \rho_\omega \ 
\]
over $M \setminus \cup_j Z_j$. Thus over this domain we get that
\[
\Big[ \,   i\Lambda_\omega F_A \, + \,  \frac{1}{\dim_\CC M} \, \big( \,\vert \phi \vert^2_h \, + \, s_\omega \, \big)  \, \Big] \, \omega \ = \ \frac{1}{\dim_\CC M}\, \Big[\, (\Lambda_{\omega'} \rho_{\omega'}) \, \omega' \, -\, (\Lambda_\omega \rho_\omega)\, \omega \, - \, \bar{s}_\omega\, \omega' \, +\, s_\omega\, \omega   \, \Big]
\]
and the result follows from the identity $\Lambda_\omega \rho_\omega = s_\omega$.
\end{proof}

\vspace{.2cm}

\medskip
\noindent
\textbf{Acknowledgements.} It is a pleasure to thank Dennis Eriksson and Nick Manton for helpful discussions about integration currents and vortex superposition, respectively. I also thank CAMGSD and Project PTDC/MAT/120411/2010 of FCT  - POPH/FSE for a generous fellowship.

%%%%%%%%%%%%%%%%%%%%%%%%%%%%%%%%%%%%%%%%%%%%%%%%%%%%%%%%%%%%%%

\vspace{.8cm}

\noindent
{\small {\textsc{Centre for Mathematical Analysis, Geometry, and Dynamical Systems (CAMGSD), Instituto Superior T\'ecnico, Av. Rovisco Pais, 1049-001 Lisbon, Portugal} }}

\noindent
{\it Email address:}
{\tt joao.o.baptista@gmail.com}

\end{document}